\documentclass[onecolumn, draftcls, 11pt]{IEEEtran}
\usepackage[cmex10]{amsmath}
\usepackage{amssymb}
\usepackage{amsthm}
\newtheorem{theorem}{Theorem}

\newtheorem{lemma}{Lemma}
\theoremstyle{remark}

\usepackage{epsfig}
\usepackage{color}

\DeclareMathOperator*{\Var}{Var}
\DeclareMathOperator*{\Cov}{Cov}

\begin{document}
\title{On the Vacationing CEO Problem: Achievable Rates and Outer Bounds}
\author{\IEEEauthorblockN{Rajiv Soundararajan$^*$, Aaron B. Wagner$^\dagger$ and Sriram Vishwanath$^*$\\}
\IEEEauthorblockA{$^*$Department of ECE, The University of Texas at Austin, Austin, TX 78712, USA\\
$^\dagger$School of ECE, Cornell University, Ithaca, NY 14853, USA
}
}

\maketitle

\begin{abstract}
This paper studies a class of source coding problems that combines elements of the CEO problem with the multiple description problem. In this setting, noisy versions of one remote source are observed by two nodes with encoders (which is similar to the CEO problem). However, it differs from the CEO problem in that each node must generate multiple descriptions of the source. This problem is of interest in multiple scenarios in efficient communication over networks. 
In this paper, an achievable region and an outer bound are presented for this problem, which is shown to be sum rate optimal for a class of distortion constraints.
 
\end{abstract}

\section{Introduction}

Determining the theoretical limits of lossy compression schemes are of significant interest. Results in lossy source coding have applications in multiple domains including  multimedia communication \& storage, image processing and distributed processing over sensor networks. 
A single representation for a single source is today a fairly well established field of research \cite{Cover}. When multiple representations and/or sources are involved, there are only a limited set of exact results known. The lossless compression of correlated sources, studied in \cite{Slepian1973} by Slepian and Wolf, is one of the early success stories in this domain. Subsequently, the Gaussian two-terminal multiple description (MD) rate region was characterized in \cite{Ozarow1980}. More recently, many new results have emerged in the field of Gaussian multiterminal source coding \cite{Oohama1997}, \cite{Krithivasan2007}. In particular, the Gaussian CEO problem was studied in \cite{Oohama1998} and \cite{Prabhakaran2004}, where the sum rate and the entire rate region were characterized.
\cite{Wang2009} provides a simplified converse argument for the sum rate. The rate region of the Gaussian two-encoder problem was characterized in \cite{Wagner2006}.

We consider a version of the CEO problem in which the CEO can  ``vacation.'' The setup is described by Figure 1.  We have a single source $S$, and two corrupted versions of the source $X_1 = S + N_1$ and $X_2 = S + N_2$ are available at the two encoders in the system. The encoders wish to communicate information about $S$ to a decoder, i.e., the CEO, which they accomplish by each sending a data packet at time 1 and another at time 2. The CEO may be on vacation during time 1, time 2, neither, or both, and she cannot receive data packets when she is vacationing.  We assume that the CEO's holiday schedule is unknown to the encoders. If the CEO works during time 1 and vacations during time 2, she expects to reproduce the source $S$ to distortion $D_{1}$. Likewise, if she works during time 2 and vacations during time 1, she expects to reproduce $S$ to distortion $D_2$. If the CEO eschews vacation and works during both periods, then she expects to reproduce $S$ to distortion $D_0$. For convenience, we represent the three vacation states of the CEO by three separate receivers in Figure 1. 
Details on the system model and problem at hand are presented in Section \ref{sec:probdef}. 

This problem generalizes both the CEO problem, by omitting the transmission at time 2, and the MD problem, by omitting the noises $N_1$ and $N_2$.
We note that a related problem is considered in \cite{Chen2008}, which also generalizes both CEO and MD. However, unlike the other problem formulation, we are able to obtain a more conclusive sum rate result. The vacationing-CEO problem arises in multicast networks in which receivers enter and depart the systems at arbitrary times. See \cite{Ahmed:ISIT09,Ahmed:ITW09} for additional discussion of this connection.

The rest of this paper is organized as follows: In the next section, we state the problem and describe the main result of the paper, which characterizes the sum rate of the Gaussian problem with CEO vacations. In Section \ref{sec:ach}, we show that the sum rate described in Section \ref{sec:probdef} is achievable, and in Section \ref{sec:ob}, we provide a lower bound for the Gaussian vacationing-CEO problem. This lower bound combines converse techniques developed individually for the MD problem~\cite{Ozarow1980,Wang2007}  and the CEO problem~\cite{Oohama1998,Wang2009}. In fact, it is interesting to note that our lower bound requires the use of \emph{both} converse techniques for the CEO problem, as neither alone is sufficient.
In Section \ref{sec:comp}, we establish the equivalence of the achievable sum rate and the lower bound on it, thus proving the main result of the paper.

\subsection{Notation}
We use capital letters to denote random variables and $\mathbb{E}[S]$ to denote the expected value of a random variable $S$. All logarithms used in the paper are natural logarithms. $\Var(S|T)$ denotes $\mathbb{E}_{S,T}[(S-\mathbb{E}[S|T])^{2}]$. For $\bar{S}=(S_{1},S_{2})$, $\Cov(\bar{S}|T)$ denotes the matrix $\mathbb{E}_{\bar{S},T}[(\bar{S}-\mathbb{E}[\bar{S}|T])(\bar{S}-\mathbb{E}[\bar{S}|T])^{\dagger}]$, where $\bar{S}^{\dagger}$ is the transpose of the vector $\bar{S}$.

\section{Problem Statement and Main Result}
\label{sec:probdef}

\begin{figure}[!th]\label{fig:model}
\centering
\scalebox{0.5}{
\input{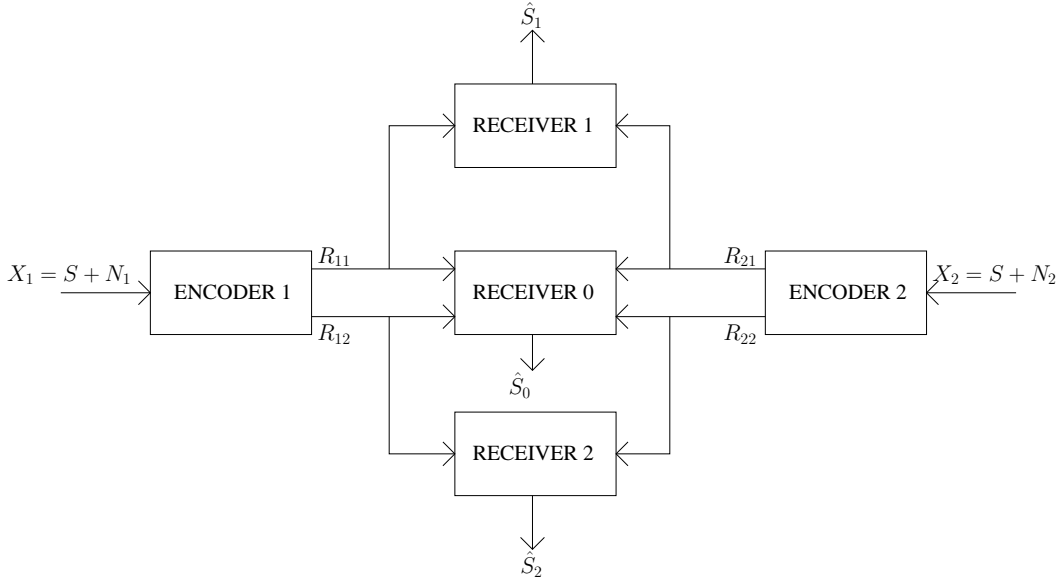}}
\caption{System Model}
\end{figure}
Let $\{X_{1i}\}_{i=1}^{n}$ and $\{X_{2i}\}_{i=1}^{n}$ be noisy observations of an underlying Gaussian source $\{S_{i}\}_{i=1}^{n}$, observed by two different encoders. The observations and the source are assumed to be independent and identically distributed (i.i.d.) over $i$. For each time instant $i$, the observations are given by 
\begin{align*}
X_{1i} = S_{i}+N_{1i}\\
X_{2i} = S_{i}+N_{2i}
\end{align*}
where $N_{1i}$ and $N_{2i}$ are Gaussian distributed with mean zero and variance $\sigma_{N_{1}}^{2}$ and $\sigma_{N_{2}}^{2}$. $S_{i}$ has mean zero and variance $\sigma_{S}^{2}$. 
Encoder $k$ observes $\{X_{ki}\}_{i=1}^{n}$ for $k=1,2$ and sends two descriptions given by $C_{kl}=f_{kl}(X_{k}^{n})$, for $l=1,2$ to two receivers. Let $R_{kl}\geq\frac{1}{n}\log |C_{kl}|$. Receiver $l$, gets the messages $f_{1l}(X_{1}^{n})$ and $f_{2l}(X_{2}^{n})$, and applies decoding function $\varphi_{l}^{n}(f_{1l}(X_{1}^{n}),f_{2l}(X_{2}^{n}))$ to obtain an estimate of the source $S^{n}$ , denoted by $\hat{S}_{l}^{n}$.The central receiver gets all the four descriptions and applies the function  $\varphi_{0}^{n}(f_{11}(X_{1}^{n}),f_{21}(X_{2}^{n}),f_{12}(X_{1}^{n}),f_{22}(X_{2}^{n}))$ to get $\hat{S}_{0}^{n}$. We say that the tuple $(R_{11},R_{12},R_{21},R_{22},D_{1},D_{2},D_{0})$ is \emph{achievable} if there exist encoding and decoding functions such that 
\begin{align*}
D_{l} &\geq \frac{1}{n}\sum_{i=1}^{n}\mathbb{E}[(S_{i}-\hat{S}_{li})^{2}], \quad l\in\{0,1,2\}.
\end{align*}

We now state the main result of the paper. Let 
\begin{align*}
\mathcal{U} = \{&(U_{11},U_{12},U_{21},U_{22}): U_{kl} = X_{k}+W_{kl} \textrm{ for } k,l\in\{1,2\}, W_{kl}\sim \mathcal{N}(0,\sigma_{W_{kl}}^{2}) ,\\
& (U_{11},U_{12})-X_{1}-X_{2}-(U_{21},U_{22}),\\
& \mathbb{E}[(S-\mathbb{E}[S|U_{1l},U_{2l}])^{2}] \leq D_{l} \textrm{ for } l \in\{1,2\} \textrm{ and }\\
& \mathbb{E}[(S-\mathbb{E}[S|U_{11},U_{12},U_{21},U_{22}])^{2}] \leq D_{0}\}.
\end{align*}
\begin{theorem}
The sum rate of the vacationing-CEO problem with distortion constraints $(D_{1},D_{2},D_{0})$ such that $\frac{1}{D_{1}}+\frac{1}{D_{2}}-\max\{\frac{1}{\sigma_{N_{1}}^{2}},\frac{1}{\sigma_{N_{2}}^{2}}\}-\frac{1}{\sigma_{S}^{2}}\geq \frac{1}{D_{0}}$ is given by
\begin{equation*}
\inf_{(U_{11},U_{12},U_{21},U_{22})\in\mathcal{U}} I(X_{1},X_{2};U_{11},U_{12},U_{21},U_{22})+I(U_{11},U_{21};U_{12},U_{22}).
\end{equation*}
\end{theorem}

The technical condition on the distortions implies that the distortion constraint at the central receiver satisfies $\max\{\frac{1}{D_{1}},\frac{1}{D_{2}}\}\leq\frac{1}{D_{0}}\leq\frac{1}{D_{1}}+\frac{1}{D_{2}}-\max\{\frac{1}{\sigma_{N_{1}}^{2}},\frac{1}{\sigma_{N_{2}}^{2}}\}-\frac{1}{\sigma_{S}^{2}}$. This means that the central distortion constraint is comparable with (although lesser than) the distortion constraints at the individual receivers. Note that the sum rate achieved here is a more general version of the achievable sum rate of the CEO problem with two sensors and the MD problem with two descriptions. In effect, the vacationing-CEO problem with just one sensor and encoder (or $R_{21}=0$ and $R_{22}=0$) is a remote source version of the two description problem. In the absence of the central receiver (or $D_{0}\geq\sigma_{S}^{2}$), the vacationing-CEO problem reduces to two CEO problems corresponding to Receivers 1 and 2. We discuss the achievability in the following section.

\section{Achievability}\label{sec:ach}
The achievable scheme discussed below is a Gaussian scheme. We define auxiliaries, $U_{11}$,$U_{12}$,$U_{21}$ and $U_{22}$ such that
\begin{align*}
U_{11} &= X_{1}+W_{11}\\
U_{12} &= X_{1}+W_{12}\\
U_{21} &= X_{2}+W_{21}\\
U_{22} &= X_{2}+W_{22},
\end{align*}
where the vector $\mathbf{W}=(W_{11},W_{12},W_{21},W_{22})$ is Gaussian distributed with mean zero and covariance matrix
\begin{equation*}
\mathbf{K}_{w} = \left[\begin{array}{ccccc} \sigma_{W_{11}}^{2} & -a_{1} & 0 & 0 \\
-a_{1} & \sigma_{W_{12}}^{2} & 0 & 0 \\
0 & 0 & \sigma_{W_{21}}^{2} & -a_{2}\\
0 & 0 & -a_{2} & \sigma_{W_{22}}^{2}
\end{array}\right].
\end{equation*}
$\mathbf{K}_{w}$ is appropriately chosen to meet the distortion constraints at Receivers 1 and 2 and the central receiver. In effect $\mathbf{K}_{w}$ is chosen such that 
\begin{align*}
\mathbb{E}\left[\left(S-\mathbb{E}\left[S|U_{1l},U_{2l}\right]\right)^{2}\right] &\leq D_{l}, l\in\{1,2\}\\
\mathbb{E}\left[\left(S-\mathbb{E}\left[S|U_{11},U_{12},U_{21},U_{22}\right]\right)^{2}\right] &\leq D_{0}.
\end{align*}

\subsection{Codebook Generation}
Encoder $k$, $k=1,2$, generates $2^{nR_{k1}'}$ $U_{k1}^{n}$ and $2^{nR_{k2}'}$ $U_{k2}^{n}$ such that $U_{k1i}$ and $U_{k2i}$ are generated i.i.d. according to the marginal of $U_{k1}$ and $U_{k2}$ respectively. $2^{nR_{k1}'}$ $U_{k1}^{n}$ and $2^{nR_{k2}'}$ $U_{k2}^{n}$ are binned into $2^{nR_{k1}}$ and $2^{nR_{k2}}$ bins respectively. 

\subsection{Encoding}
Encoder $k$ chooses the pair $(U_{k1}^{n},U_{k2}^{n})$ jointly typical with $X_{k}^{n}$ and transmits the respective bin indexes. There exists a pair $(U_{k1}^{n},U_{k2}^{n})$ jointly typical with $X_{k}^{n}$ with high probability if 
\begin{align}
R_{k1}' &> I(X_{k};U_{k1})\nonumber\\
R_{k2}' &> I(X_{k};U_{k2})\label{eqn:generrors}\\
R_{k1}'+R_{k2}' &> I(X_{k};U_{k1},U_{k2})+I(U_{k1};U_{k2}).  \nonumber
\end{align}
This multiple description encoding scheme is similar to the scheme in \cite{Cover1982}. Since $(U_{11}^{n},U_{12}^{n})-X_{1}^{n}-X_{2}^{n}-(U_{21}^{n},U_{22}^{n})$, by the Markov lemma (Lemma 14.8.1) in \cite{Cover}, we also have that $(U_{11}^{n},U_{12}^{n},U_{21}^{n},U_{22}^{n})$ are jointly typical.

\subsection{Decoding at individual receivers}
Receiver $l$, $l=1,2$, looks for $U_{1l}^{n}$ and $U_{2l}^{n}$ that are jointly typical in the bins corresponding to the bin indexes it receives. Receiver $l$  will be able to find unique codewords $U_{1l}^{n}$ and $U_{2l}^{n}$ that are jointly typical if
\begin{align}
R_{1l} &> R_{1l}'-I(U_{1l};U_{2l})\nonumber\\
R_{2l} &> R_{2l}'-I(U_{1l};U_{2l})\label{eqn:decerrors}\\
R_{1l}+R_{2l} &> R_{1l}'+R_{2l}'-I(U_{1l};U_{2l}). \nonumber
\end{align}
Receiver $l$ generates an estimate of $S^{n}$, by constructing the minimum mean squared estimate (MMSE)  $\mathbb{E}[S^{n}|U_{1l}^{n},U_{2l}^{n}]$. The decoding scheme resembles the decoding in the Berger-Tung scheme \cite{Tung}. 

\subsection{Decoding at central receiver}
Receiver 0 mimics the decoding at Receiver 1 and 2 to find jointly typical pairs $(U_{11}^{n},U_{21}^{n})$ and $(U_{12}^{n},U_{22}^{n})$ in the received bin indexes. Therefore, Receiver 0 will be able to find such unique codewords if the rates satisfy (\ref{eqn:decerrors}). Since $(U_{11}^{n},U_{12}^{n},U_{21}^{n},U_{22}^{n})$ are jointly typical, Receiver 0 constructs the MMSE estimate of $S^{n}$ given by $\mathbb{E}[S^{n}|U_{11}^{n},U_{12}^{n},U_{21}^{n},U_{22}^{n}]$.

Note that the equations in (\ref{eqn:generrors}) and (\ref{eqn:decerrors}) represent the entire rate region achievable by the Gaussian scheme for the vacationing-CEO problem. We now consider the sum rate achievable by the Gaussian scheme. 
\begin{lemma}\label{thm:achsum}
The sum rate achievable by the Gaussian scheme is given by 
\begin{equation}
\inf_{(U_{11},U_{12},U_{21},U_{22})\in\mathcal{U}}I(X_{1},X_{2};U_{11},U_{12},U_{21},U_{22})+I(U_{11},U_{21};U_{12},U_{22}). \label{eqn:ach_sum}
\end{equation}
\end{lemma}
The lemma is proved in Appendix \ref{sec:lemachsum}. We present the lower bound on the sum rate in the next section.

\section{Lower Bound}\label{sec:ob}
We now make a few definitions before presenting the lower bound on the sum rate. $C_{kl}$ denotes the code from Encoder $k$ to Receiver $l$ for $k=1,2$ and $l=1,2$.

Define, 
\begin{align}
d_{11} &= \frac{1}{n}\sum_{i=1}^{n}\Var(X_{1i}|C_{11},S^{n})  &d_{21} &= \frac{1}{n}\sum_{i=1}^{n}\Var(X_{2i}|C_{21},S^{n}) \nonumber\\
d_{12} &= \frac{1}{n}\sum_{i=1}^{n}\Var(X_{1i}|C_{12},S^{n})  &d_{22} &= \frac{1}{n}\sum_{i=1}^{n}\Var(X_{2i}|C_{22},S^{n}) \label{eqn:param}\\
t_{1} &= \frac{1}{n}I(X_{1}^{n};C_{11},C_{12}|S^{n})  &t_{2} &= \frac{1}{n}I(X_{2}^{n};C_{21},C_{22}|S^{n}).\nonumber
\end{align}
We remark that in the following $0<D_{0}<\min\{D_{1},D_{2}\}$ and $\max\{D_{1},D_{2}\}<\sigma_{S}^{2}$. We now define for $k=1,2$, 
\begin{align*}
\mathcal{F}_{k} = \{&(d_{1},d_{2},t): d_{1},d_{2},t \in [0,\infty)\\ 
&\sigma_{N_{k}}^{2}e^{-2t}\leq \min\{d_{1},d_{2}\} \quad \max\{d_{1},d_{2}\} \leq \sigma_{N_{k}}^{2}\}.
\end{align*}
Further define, 
\begin{align}
\mathcal{F} = \{&(d_{11},d_{12},d_{21},d_{22},t_{1},t_{2}) :(d_{k1},d_{k2},t_{k})\in\mathcal{F}_{k}, k=1,2\nonumber\\
&\frac{1}{D_{1}} \leq \frac{1}{\sigma_{S}^{2}}+\frac{1}{\sigma_{N_{1}}^{2}}+\frac{1}{\sigma_{N_{2}}^{2}}-\frac{d_{11}}{\sigma_{N_{1}}^{4}}-\frac{d_{21}}{\sigma_{N_{2}}^{4}}\label{eqn:d1con}\\
&\frac{1}{D_{2}} \leq
\frac{1}{\sigma_{S}^{2}}+\frac{1}{\sigma_{N_{1}}^{2}}+\frac{1}{\sigma_{N_{2}}^{2}}-\frac{d_{12}}{\sigma_{N_{1}}^{4}}-\frac{d_{22}}{\sigma_{N_{2}}^{4}}\label{eqn:d2con}\\
&\frac{1}{D_{0}}\leq \frac{1}{\sigma_{S}^{2}} + \frac{1-e^{-2t_{1}}}{\sigma_{N_{1}}^{2}} + \frac{1-e^{-2t_{2}}}{\sigma_{N_{2}}^{2}}\label{eqn:d0con}\}.
\end{align} 
We have the following lemma which characterizes the parameters $\bar{p} = (d_{11},d_{12},d_{21},d_{22},t_{1},t_{2})$ defined above. 
\begin{lemma}\label{thm:lemparam}
The parameters defined in (\ref{eqn:param}) satisfy $\bar{p}\in\mathcal{F}$.
\end{lemma}
\begin{proof}
The proof that 
\begin{equation*}
\frac{1}{D_{0}}\leq \frac{1}{\sigma_{S}^{2}} + \frac{1-e^{-2t_{1}}}{\sigma_{N_{1}}^{2}} + \frac{1-e^{-2t_{2}}}{\sigma_{N_{2}}^{2}}
\end{equation*}
follows directly from Lemma 3.1 in \cite{Prabhakaran2004}. Also, in Theorem 1 in \cite{Wang2009}, it is shown that
\begin{equation*}
\frac{1}{D_{l}} \leq
\frac{1}{\sigma_{S}^{2}}+\frac{1}{\sigma_{N_{1}}^{2}}+\frac{1}{\sigma_{N_{2}}^{2}}-\frac{d_{1l}}{\sigma_{N_{1}}^{4}}-\frac{d_{2l}}{\sigma_{N_{2}}^{4}}
\end{equation*}
for $l=1,2$. 
By definition, 
\begin{align*}
nt_{k}=I(X_{k}^{n},C_{k1},C_{k2}|S^{n})&=h(X_{k}^{n}|S^{n})-h(X_{k}^{n}|C_{k1},C_{k2},S^{n})\\
&\geq \frac{n}{2}\log \sigma_{N_{k}}^{2}-h(X_{k}^{n}|C_{kl},S^{n}), l=1,2\\
&\geq \frac{n}{2}\log \sigma_{N_{k}}^{2}-\frac{n}{2}\log d_{kl}, l=1,2.
\end{align*}
Therefore for $k=1,2$,
\begin{equation*}
\sigma_{N_{k}}^{2}e^{-2t_{k}} \leq \min\{d_{k1},d_{k2}\}. 
\end{equation*}
Also, since $\mathbb{E}[N_{k}^{n}|C_{kl}]$ achieves a smaller mean squared error in $N_{k}^{n}$ than any other estimator, 
\begin{equation*}
d_{kl} = \frac{1}{n}\sum_{i=1}^{n}\Var(X_{ki}|C_{kl},S^{n})=\frac{1}{n}\sum_{i=1}^{n}\Var(N_{ki}|C_{kl}) \leq \sigma_{N_{k}}^{2}
\end{equation*}
for $k=1,2$ and $l=1,2$. 
This concludes the proof of the lemma. 
\end{proof}
Define, 
\begin{align}
\mathcal{P}_{1} = \{&(d_{11},d_{12},d_{21},d_{22},t_{1},t_{2}) :(d_{11},d_{12},d_{21},d_{22},t_{1},t_{2})\in \mathcal{F}\nonumber\\
&\frac{1}{D_{1}} = \frac{1}{\sigma_{S}^{2}}+\frac{1}{\sigma_{N_{1}}^{2}}+\frac{1}{\sigma_{N_{2}}^{2}}-\frac{d_{11}}{\sigma_{N_{1}}^{4}}-\frac{d_{21}}{\sigma_{N_{2}}^{4}}\label{eqn:d1eq}\\
&\frac{1}{D_{2}} =
\frac{1}{\sigma_{S}^{2}}+\frac{1}{\sigma_{N_{1}}^{2}}+\frac{1}{\sigma_{N_{2}}^{2}}-\frac{d_{12}}{\sigma_{N_{1}}^{4}}-\frac{d_{22}}{\sigma_{N_{2}}^{4}}\label{eqn:d2eq}\\
&\frac{1}{D_{0}}= \frac{1}{\sigma_{S}^{2}} + \frac{1-e^{-2t_{1}}}{\sigma_{N_{1}}^{2}} + \frac{1-e^{-2t_{2}}}{\sigma_{N_{2}}^{2}}\label{eqn:d0eq}\}\\
\mathcal{P}_{2} = \{&(d_{11},d_{12},d_{21},d_{22},t_{1},t_{2}) :(d_{11},d_{12},d_{21},d_{22},t_{1},t_{2})\in \mathcal{F}\nonumber\\
&\frac{1}{D_{1}} = \frac{1}{\sigma_{S}^{2}}+\frac{1}{\sigma_{N_{1}}^{2}}+\frac{1}{\sigma_{N_{2}}^{2}}-\frac{d_{11}}{\sigma_{N_{1}}^{4}}-\frac{d_{21}}{\sigma_{N_{2}}^{4}}\nonumber\\
&\frac{1}{D_{2}} =
\frac{1}{\sigma_{S}^{2}}+\frac{1}{\sigma_{N_{1}}^{2}}+\frac{1}{\sigma_{N_{2}}^{2}}-\frac{d_{12}}{\sigma_{N_{1}}^{4}}-\frac{d_{22}}{\sigma_{N_{2}}^{4}}\nonumber\\
&\frac{1}{D_{0}}< \frac{1}{\sigma_{S}^{2}} + \frac{1-e^{-2t_{1}}}{\sigma_{N_{1}}^{2}} + \frac{1-e^{-2t_{2}}}{\sigma_{N_{2}}^{2}}\nonumber\\
&\sigma_{N_{1}}^{2}e^{-2t_{1}}= \min\{d_{11},d_{12}\} \quad \sigma_{N_{2}}^{2}e^{-2t_{2}}= \min\{d_{21},d_{22}\}\}\nonumber .
\end{align} 
We denote
\begin{equation*}
\mathcal{P} = \mathcal{P}_{1}\cup \mathcal{P}_{2}.
\end{equation*}
Note that the definition of $\mathcal{P}$ imposes the restriction on the parameters to satisfy the individual distortion constraints with equality. The central distortion constraint may be satisfied with equality or the parameters satisfy $\sigma_{N_{k}}^{2}e^{-2t_{k}}=\min\{d_{k1},d_{k2}\}$ for $k=1,2$. We also observe that $\mathcal{P}\subset\mathcal{F}$.\\
Let $\bar{p}\in\mathcal{F}$. Then $\Delta\mathcal{F}_{\bar{p}}$ is defined as 
\begin{align*}
\Delta\mathcal{F}_{\bar{p}} = \{&\Delta\bar{p}=(\Delta d_{11},\Delta d_{12},\Delta d_{21},\Delta d_{22},-\Delta t_{1},-\Delta t_{2}):\Delta d_{11},\Delta d_{12},\Delta d_{21},\Delta d_{22},\Delta t_{1},\Delta t_{2}\in[0,\infty) \textrm{ and }\\ &(d_{11}+\Delta d_{11},d_{12}+\Delta d_{12},d_{21}+\Delta d_{21},d_{22}+\Delta d_{22},t_{1}-\Delta t_{1},t_{2}-\Delta t_{2}) \in \mathcal{P}\}.
\end{align*}
\begin{lemma}\label{thm:deltarem}
$\Delta\mathcal{F}_{\bar{p}}\neq \phi  \quad \forall \bar{p}\in\mathcal{F}$. 
\end{lemma}
\begin{proof}
The lemma is proved as follows. Consider $\bar{p}\in\mathcal{F}$. Then we increase $d_{11}$ and $d_{12}$ by $\Delta d_{11}$ and $\Delta d_{12}$ until we meet the distortion constraints at individual receivers with equality or $d_{1l}+\Delta d_{1l}=\sigma_{N_{1}}^{2}$, $l=1,2$. In the former case, we satisfy the individual distortion constraints with equality. In the latter case, we now increase $d_{21}$ and $d_{22}$ by $\Delta d_{21}$ and $\Delta d_{22}$ until we meet the individual distortion constraints with equality. We will be able to find such $\Delta d_{21}$ and $\Delta d_{22}$ satisfying $d_{2l}+\Delta d_{2l}\leq\sigma_{N_{2}}^{2}$, $l=1,2$, since $D_{l}<\sigma_{S}^{2}$ for $l=1,2$. Now, we decrease $t_{1}$ by $\Delta t_{1}$ until the central distortion constraint is met with equality or $\sigma_{N_{1}}^{2}e^{-2(t_{1}-\Delta t_{1})}=\min\{d_{11}+\Delta d_{11},d_{12}+\Delta d_{12}\}$. In the former case, we satisfy the central distortion with equality. In the latter case, we decrease $t_{2}$ by $\Delta t_{2}$ until the central distortion constraint is met with equality or $\sigma_{N_{2}}^{2}e^{-2(t_{2}-\Delta t_{2})}=\min\{d_{21}+\Delta d_{21},d_{22}+\Delta d_{22}\}$. Therefore $\forall \bar{p}\in\mathcal{F}$, $\Delta\mathcal{F}_{\bar{p}}\neq \phi$. 
\end{proof}
Let $k=1,2$ and $\sigma_{Z}^{2}\geq 0$. Define, for $(d_{1},d_{2},t)\in\mathcal{F}_{k}$, 
\begin{equation}\label{eqn:rparam}
r_{k}(d_{1},d_{2},t,\sigma_{Z}^{2}) = t+\frac{1}{2}\log\frac{(\sigma_{N_{k}}^{2}+\sigma_{Z}^{2})}{(d_{1}+\sigma_{Z}^{2})(d_{2}+\sigma_{Z}^{2})} + \frac{1}{2}\log(\sigma_{N_{k}}^{2}e^{-2t}+\sigma_{Z}^{2}). 
\end{equation}
We now state the main result of this section. 

\begin{lemma}\label{thm:lb}
The sum rate of the vacationing-CEO problem is lower bounded by 
\begin{align}
\inf_{\bar{p}\in\mathcal{P}} \sup_{\sigma_{Z_{1}},\sigma_{Z_{2}}\in\mathbb{R}}r_{1}(d_{11},d_{12},t_{1},\sigma_{Z_{1}}^{2}) + r_{2}(d_{21},d_{22},t_{2},\sigma_{Z_{2}}^{2})+\frac{1}{2}\log\frac{\sigma_{S}^{4}}{D_{1}D_{2}}.
\label{eqn:lb}
\end{align} 
\end{lemma}

\begin{proof}
By procedural steps, we have  
\begin{align}
n(R_{11}+R_{21}+R_{12}+R_{22}) \geq& H(C_{11},C_{21})+H(C_{12},C_{22})\nonumber\\
\geq& H(C_{11},C_{21})+H(C_{12},C_{22})-H(C_{11},C_{21},C_{12},C_{22})\nonumber\\
& +H(C_{11},C_{21},C_{12},C_{22})-H(C_{11},C_{21},C_{12},C_{22}|X_{1}^{n},X_{2}^{n})\nonumber\\
=& I(X_{1}^{n},X_{2}^{n};C_{11},C_{21},C_{12},C_{22})+I(C_{11},C_{21};C_{12},C_{22})\nonumber\\
\overset{(a)}{=}& I(S^{n};C_{11},C_{21},C_{12},C_{22})+I(X_{1}^{n},X_{2}^{n};C_{11},C_{21},C_{12},C_{22}|S^{n})\nonumber\\
 & +I(C_{11},C_{21};C_{12},C_{22})\nonumber\\
\overset{(b)}{=}&  I(S^{n};C_{11},C_{21},C_{12},C_{22}) + I(X_{1}^{n};C_{11},C_{12}|S^{n})+I(X_{2}^{n};C_{21},C_{22}|S^{n})\nonumber\\
  &+I(C_{11},C_{21};C_{12},C_{22}), \label{eqn:step0}
\end{align}
where (a) is true since $S^{n}-(X_{1}^{n},X_{2}^{n})-(C_{11},C_{12},C_{21},C_{22})$ and (b) is true since $(C_{11},C_{12})-X_{1}^{n}-S^{n}-X_{2}^{n}-(C_{21},C_{22})$. \\
Let $Y_{1i}=X_{1i}+Z_{1i}$ and $Y_{2i}=X_{2i}+Z_{2i}$,  where $Z_{1i}$ and $Z_{2i}$ are i.i.d Gaussians with mean zero and variance $\sigma_{Z_{1}}^{2}$ and $\sigma_{Z_{2}}^{2}$ for $i\in\{1,2,\ldots,n\}$. Also, $Z_{1i}$ and $Z_{2i}$ are independent of $S_{i}$, $X_{1i}$ and $X_{2i}$. Now,
\begin{align}
I(C_{11},C_{21};C_{12},C_{22}) =& H(C_{11},C_{21})+H(C_{12},C_{22})-H(C_{11},C_{21},C_{12},C_{22})\nonumber\\
=& H(C_{11},C_{21})+H(C_{12},C_{22})-H(C_{11},C_{21},C_{12},C_{22})\nonumber\\
& -H(C_{11},C_{21}|S^{n},Y_{1}^{n},Y_{2}^{n})-H(C_{12},C_{22}|S^{n},Y_{1}^{n},Y_{2}^{n})\nonumber\\
& +H(C_{11},C_{21},C_{12},C_{22}|S^{n},Y_{1}^{n},Y_{2}^{n}) + I(C_{11},C_{21};C_{12},C_{22}|Y_{1}^{n},Y_{2}^{n},S^{n})\nonumber\\
=& I(S^{n},Y_{1}^{n},Y_{2}^{n};C_{11},C_{21})+I(S^{n},Y_{1}^{n},Y_{2}^{n};C_{12},C_{22})\nonumber\\
& -I(S^{n},Y_{1}^{n},Y_{2}^{n};C_{11},C_{12},C_{21},C_{22})+I(C_{11},C_{21};C_{12},C_{22}|Y_{1}^{n},Y_{2}^{n},S^{n})\nonumber\\
\geq & I(S^{n},Y_{1}^{n},Y_{2}^{n};C_{11},C_{21})+I(S^{n},Y_{1}^{n},Y_{2}^{n};C_{12},C_{22})\nonumber\\
& -I(S^{n},Y_{1}^{n},Y_{2}^{n};C_{11},C_{12},C_{21},C_{22}). \label{eqn:step1}
\end{align}
For $l=1,2$,
\begin{equation*}
I(S^{n},Y_{1}^{n},Y_{2}^{n};C_{1l},C_{2l}) = I(S^{n};C_{1l},C_{2l})+I(Y_{1}^{n};C_{1l}|S^{n})+I(Y_{2}^{n};C_{2l}|S^{n})
\end{equation*}
since $(Y_{1}^{n},C_{1l})-S^{n}-(Y_{2}^{n},C_{2l})$. By the definition of the rate distortion function for Gaussian random variables, $I(S^{n};C_{1l},C_{2l})\geq\frac{n}{2}\log\frac{\sigma_{S}^{2}}{D_{l}}$ and $I(Y_{k}^{n};C_{kl}|S^{n})\geq\frac{n}{2}\log\frac{\sigma_{N_{k}}^{2}+\sigma_{Z_{k}}^{2}}{d_{kl}+\sigma_{Z_{k}}^{2}}$ for $k=1,2$. Therefore, 
\begin{equation}\label{eqn:step2}
I(S^{n},Y_{1}^{n},Y_{2}^{n};C_{1l},C_{2l}) \geq \frac{n}{2}\log \frac{\sigma_{S}^{2}(\sigma_{N_{1}}^{2}+\sigma_{Z_{1}}^{2})(\sigma_{N_{2}}^{2}+\sigma_{Z_{2}}^{2})}{D_{l}(d_{1l}+\sigma_{Z_{1}}^{2})(d_{2l}+\sigma_{Z_{2}}^{2})}.
\end{equation}
Observe that
\begin{align}
I(S^{n},Y_{1}^{n},Y_{2}^{n};C_{11},C_{12},C_{21},C_{22}) =& I(S^{n};C_{11},C_{21},C_{12},C_{22})+I(Y_{1}^{n},Y_{2}^{n};C_{11},C_{21},C_{12},C_{22}|S^{n})\nonumber\\
=& I(S^{n};C_{11},C_{21},C_{12},C_{22})+I(Y_{1}^{n};C_{11},C_{12}|S^{n})\nonumber\\
&+I(Y_{2}^{n};C_{21},C_{22}|S^{n}), \label{eqn:step3}
\end{align}
where in the last step we used $(Y_{1}^{n},C_{11},C_{12})-S^{n}-(Y_{2}^{n},C_{21},C_{22})$. Further, for $k=1,2$
\begin{align}
I(Y_{k}^{n};C_{k1},C_{k2}|S^{n}) &= -h(Y_{k}^{n}|S^{n},C_{k1},C_{k2})+h(Y_{k}^{n}|S^{n})\nonumber\\
&\overset{(c)}{\leq} -\frac{n}{2}\log(e^{\frac{2}{n}h(X_{k}^{n}|S^{n},C_{k1},C_{k2})}+e^{\frac{2}{n}h(Z_{k}^{n})})+h(Y_{k}^{n}|S^{n})\nonumber\\
&= -\frac{n}{2}\log(e^{\frac{2}{n}(h(X_{k}^{n}|S^{n})-I(X_{k}^{n};C_{k1},C_{k2}|S^{n}))}+e^{\frac{2}{n}h(Z_{k}^{n})})+h(Y_{k}^{n}|S^{n})\nonumber\\
&= -\frac{n}{2}\log(\sigma_{N_{k}}^{2}e^{-2t_{k}}+\sigma_{Z_{k}}^{2})+\frac{n}{2}\log(\sigma_{N_{k}}^{2}+\sigma_{Z_{k}}^{2}), \label{eqn:step4}
\end{align}
where (c) follows from entropy power inequality (EPI). From (\ref{eqn:step1}), (\ref{eqn:step2}), (\ref{eqn:step3}) and (\ref{eqn:step4}),
\begin{align*}
I(C_{11},C_{21};C_{12},C_{22}) \geq & \frac{n}{2}\log\frac{(\sigma_{N_{1}}^{2}+\sigma_{Z_{1}}^{2})(\sigma_{N_{2}}^{2}+\sigma_{Z_{2}}^{2})\sigma_{S}^{4}}{(d_{12}+\sigma_{Z_{1}}^{2})(d_{22}+\sigma_{Z_{2}}^{2})(d_{11}+\sigma_{Z_{1}}^{2})(d_{21}+\sigma_{Z_{2}}^{2})D_{1}D_{2}}\\
& + \frac{n}{2}\log(\sigma_{N_{1}}^{2}e^{-2t_{1}}+\sigma_{Z_{1}}^{2})(\sigma_{N_{2}}^{2}e^{-2t_{2}}+\sigma_{Z_{2}}^{2})-I(S^{n};C_{11},C_{21},C_{12},C_{22})
\end{align*}
Substituting the above in (\ref{eqn:step0}), we get
\begin{align}
R_{11}+R_{21}+R_{12}+R_{22} \geq & t_{1}+t_{2}+\frac{1}{2}\log\frac{(\sigma_{N_{1}}^{2}+\sigma_{Z_{1}}^{2})(\sigma_{N_{2}}^{2}+\sigma_{Z_{2}}^{2})\sigma_{S}^{4}}{(d_{12}+\sigma_{Z_{1}}^{2})(d_{22}+\sigma_{Z_{2}}^{2})(d_{11}+\sigma_{Z_{1}}^{2})(d_{21}+\sigma_{Z_{2}}^{2})D_{1}D_{2}}\nonumber\\
& +\frac{1}{2}\log(\sigma_{N_{1}}^{2}e^{-2t_{1}}+\sigma_{Z_{1}}^{2})(\sigma_{N_{2}}^{2}e^{-2t_{2}}+\sigma_{Z_{2}}^{2})\nonumber\\
= & r_{1}(d_{11},d_{12},t_{1},\sigma_{Z_{1}}^{2})+r_{2}(d_{21},d_{22},t_{2},\sigma_{Z_{2}}^{2})+\frac{1}{2}\log\frac{\sigma_{S}^{4}}{D_{1}D_{2}}, \label{eqn:pmset}
\end{align}
where the last equality is due to the definition in (\ref{eqn:rparam}). 
From Lemma \ref{thm:lemparam}, we have $\bar{p}\in\mathcal{F}$. By Lemma \ref{thm:deltarem}, $\Delta\mathcal{F}_{\bar{p}}\neq\phi$. Let $\Delta\bar{p}\in\Delta\mathcal{F}_{\bar{p}}$. Note that $r_{k}(d_{k1},d_{k2},t_{k},\sigma_{Z_{k}}^{2})$ is decreasing in $d_{k1}$ and $d_{k2}$ and increasing in $t_{k}$ for $k=1,2$. This implies that  
\begin{equation*}
r_{k}(d_{k1},d_{k2},t_{k},\sigma_{Z_{k}}^{2}) \geq r_{k}(d_{k1}+\Delta d_{k1},d_{k2}+\Delta d_{k2},t_{k}-\Delta t_{k},\sigma_{Z_{k}}^{2}) \quad \forall \bar{p}\in\mathcal{F}.
\end{equation*}
Therefore,
\begin{align*}
R_{11}+R_{21}+R_{12}+R_{22} \geq &r_{1}(d_{11}+\Delta d_{11},d_{12}+\Delta d_{12},t_{1}-\Delta t_{1},\sigma_{Z_{1}}^{2})\\
&+r_{2}(d_{21}+\Delta d_{21},d_{22}+\Delta d_{22},t_{2}-\Delta t_{2},\sigma_{Z_{2}}^{2})+\frac{1}{2}\log\frac{\sigma_{S}^{4}}{D_{1}D_{2}}.
\end{align*}
By definition, $\bar{p}+\Delta\bar{p}\in\mathcal{P}$. Therefore,
\begin{align*}
R_{11}+R_{21}+R_{12}+R_{22} \geq \inf_{\bar{p}\in\mathcal{P}} \sup_{\sigma_{Z_{1}},\sigma_{Z_{2}}\in\mathbb{R}} r_{1}(d_{11},d_{12},t_{1},\sigma_{Z_{1}}^{2})+r_{2}(d_{21},d_{22},t_{2},\sigma_{Z_{2}}^{2})+\frac{1}{2}\log\frac{\sigma_{S}^{4}}{D_{1}D_{2}}. 
\end{align*}
\end{proof}
In the following section, we show that the lower bound on the sum rate described above is achieved by the Gaussian scheme. 

\section{Equivalence of achievable sum rate and lower bound}\label{sec:comp}
Before we compare sum rate of the achievable scheme with the lower bound, we present two lemmas about parameters introduced in the previous section which will be used in the comparison. We will use the notation $\bar{p}=(d_{11},d_{12},d_{21},d_{22},t_{1},t_{2})$. 

\begin{lemma}\label{thm:paramlemma}
If $\frac{1}{D_{1}}+\frac{1}{D_{2}}-\max\{\frac{1}{\sigma_{N_{1}}^{2}},\frac{1}{\sigma_{N_{2}}^{2}}\}-\frac{1}{\sigma_{S}^{2}}\geq \frac{1}{D_{0}}$ and $\bar{p}\in\mathcal{P}$, then
\begin{equation*}
d_{k1}+d_{k2}-\sigma_{N_{k}}^{2}e^{-2t_{k}}-\sigma_{N_{k}}^{2} \leq 0
\end{equation*}
for $k=1,2$. 
\end{lemma}
\begin{proof}
Let $\bar{p}\in\mathcal{P}_{1}$. Since 
\begin{equation*}
\frac{1}{D_{1}}+\frac{1}{D_{2}}-\max\{\frac{1}{\sigma_{N_{1}}^{2}},\frac{1}{\sigma_{N_{2}}^{2}}\}-\frac{1}{\sigma_{S}^{2}}\geq \frac{1}{D_{0}},
\end{equation*}
substituting for $\frac{1}{D_{1}}$, $\frac{1}{D_{2}}$ and $\frac{1}{D_{0}}$, from (\ref{eqn:d1eq}), (\ref{eqn:d2eq}) and (\ref{eqn:d0eq}) respectively, we get 
\begin{equation*}
\frac{d_{11}+d_{12}-\sigma_{N_{1}}^{2}-\sigma_{N_{1}}^{2}e^{-2t_{1}}}{\sigma_{N_{1}}^{4}}+\frac{d_{21}+d_{22}-\sigma_{N_{2}}^{2}-\sigma_{N_{2}}^{2}e^{-2t_{2}}}{\sigma_{N_{2}}^{4}}+\max\{\frac{1}{\sigma_{N_{1}}^{2}},\frac{1}{\sigma_{N_{2}}^{2}}\}\leq 0.
\end{equation*}
Therefore either $d_{11}+d_{12}-\sigma_{N_{1}}^{2}-\sigma_{N_{1}}^{2}e^{-2t_{1}}\leq 0$ or $d_{21}+d_{22}-\sigma_{N_{2}}^{2}-\sigma_{N_{2}}^{2}e^{-2t_{2}}\leq 0$. Let $d_{11}+d_{12}-\sigma_{N_{1}}^{2}-\sigma_{N_{1}}^{2}e^{-2t_{1}}\leq 0$. But since  
\begin{align*}
&\frac{d_{11}+d_{12}-\sigma_{N_{1}}^{2}-\sigma_{N_{1}}^{2}e^{-2t_{1}}}{\sigma_{N_{1}}^{4}}+\frac{d_{21}+d_{22}-\sigma_{N_{2}}^{2}-\sigma_{N_{2}}^{2}e^{-2t_{2}}}{\sigma_{N_{2}}^{4}}+\frac{1}{\sigma_{N_{1}}^{2}}\leq 0\\
\Rightarrow&\frac{d_{11}+d_{12}-\sigma_{N_{1}}^{2}e^{-2t_{1}}}{\sigma_{N_{1}}^{4}}+\frac{d_{21}+d_{22}-\sigma_{N_{2}}^{2}-\sigma_{N_{2}}^{2}e^{-2t_{2}}}{\sigma_{N_{2}}^{4}} \leq 0,
\end{align*}
and $\sigma_{N_{1}}^{2}e^{-2t_{1}}\leq \min\{d_{11},d_{12}\}$, it follows that $d_{21}+d_{22}-\sigma_{N_{2}}^{2}-\sigma_{N_{2}}^{2}e^{-2t_{2}}\leq0$. Similarly, we can start with $d_{21}+d_{22}-\sigma_{N_{2}}^{2}-\sigma_{N_{2}}^{2}e^{-2t_{2}}\leq 0$, and use 
\begin{equation*}
\frac{d_{11}+d_{12}-\sigma_{N_{1}}^{2}-\sigma_{N_{1}}^{2}e^{-2t_{1}}}{\sigma_{N_{1}}^{4}}+\frac{d_{21}+d_{22}-\sigma_{N_{2}}^{2}-\sigma_{N_{2}}^{2}e^{-2t_{2}}}{\sigma_{N_{2}}^{4}}+\frac{1}{\sigma_{N_{2}}^{2}}\leq 0,
\end{equation*}
to show that $d_{11}+d_{12}-\sigma_{N_{1}}^{2}-\sigma_{N_{1}}^{2}e^{-2t_{1}}\leq 0$. Therefore we have now shown that if $\bar{p}\in\mathcal{P}_{1}$, then $d_{k1}+d_{k2}-\sigma_{N_{k}}^{2}e^{-2t_{k}}-\sigma_{N_{k}}^{2} \leq 0$ for $k=1,2$. 

Now, let $\bar{p}\in\mathcal{P}_{2}$. Therefore, $\sigma_{N_{k}}^{2}e^{-2t_{k}}=\min\{d_{k1},d_{k2}\}$, $k=1,2$. Since $\max\{d_{k1},d_{k2}\}\leq\sigma_{N_{k}}^{2}$, it follows that 
\begin{align*}
d_{k1}+d_{k2}-\sigma_{N_{k}}^{2}-\sigma_{N_{k}}^{2}e^{-2t_{k}}&= \min\{d_{k1},d_{k2}\}+\max\{d_{k1},d_{k2}\}-\sigma_{N_{k}}^{2}-\sigma_{N_{k}}^{2}e^{-2t_{k}}\\
&= \max\{d_{k1},d_{k2}\}-\sigma_{N_{k}}^{2}\\
&\leq 0. 
\end{align*}
Thus for all $\bar{p}\in\mathcal{P}$, $d_{k1}+d_{k2}-\sigma_{N_{k}}^{2}-\sigma_{N_{k}}^{2}e^{-2t_{k}}\leq 0$, $k=1,2$. 
\end{proof}

We now state and prove the second lemma about the parameters. Let $(d_{k1},d_{k2},t_{k})\in\mathcal{F}_{k}$ for $k=1,2$. Define
\begin{equation}\label{eqn:alpha}
\alpha_{k0} = \frac{\sigma_{N_{k}}^{4}e^{-2t_{k}}}{\sigma_{N_{k}}^{2}-\sigma_{N_{k}}^{2}e^{-2t_{k}}} \quad \alpha_{k1} = \frac{\sigma_{N_{k}}^{2}d_{k1}}{\sigma_{N_{k}}^{2}-d_{k1}} \quad 
\alpha_{k2} = \frac{\sigma_{N_{k}}^{2}d_{k2}}{\sigma_{N_{k}}^{2}-d_{k2}}
\end{equation}
and
\begin{equation*}
g_{k}(\beta) = \frac{1}{\alpha_{k0}+\beta}-\frac{1}{\alpha_{k1}+\beta}-\frac{1}{\alpha_{k2}+\beta}. 
\end{equation*}
We use this function to partition the space of parameters $(d_{k1},d_{k2},t_{k})\in\mathcal{F}_{k}$. Define, 
\begin{align*}
\mathcal{F}_{k1} =& \{(d_{k1},d_{k2},t_{k})\in \mathcal{F}_{k}: g_{k}(0)>0 \textrm{ and } g_{k}(\sigma_{N_{k}}^{2})\leq 0\}  \\
\mathcal{F}_{k2} =& \{(d_{k1},d_{k2},t_{k})\in \mathcal{F}_{k}: g_{k}(0)\leq 0\}\\
\mathcal{F}_{k3} =& \{(d_{k1},d_{k2},t_{k})\in \mathcal{F}_{k}: g_{k}(\sigma_{N_{k}}^{2}) > 0\}.
\end{align*}
\begin{lemma}\label{thm:paramlem}
For $k=1,2$,
\begin{equation*}
\mathcal{F}_{k} = \mathcal{F}_{k1}\cup \mathcal{F}_{k2} \cup \mathcal{F}_{k3}.
\end{equation*}
Moreover, if $\frac{1}{D_{1}}+\frac{1}{D_{2}}-\max\{\frac{1}{\sigma_{N_{1}}^{2}},\frac{1}{\sigma_{N_{2}}^{2}}\}-\frac{1}{\sigma_{S}^{2}}\geq \frac{1}{D_{0}}$, then 
\begin{equation*}
\bar{p}\in\mathcal{P} \Rightarrow (d_{k1},d_{k2},t_{k}) \in \mathcal{F}_{k1}\cup \mathcal{F}_{k2}, k=1,2.
 \end{equation*}
\end{lemma}
\begin{proof}
For every $(d_{k1},d_{k2},t_{k})\in\mathcal{F}_{k}$, one of either $g_{k}(0)>0$ and $g_{k}(\sigma_{N_{k}}^{2})\leq 0$ or $g_{k}(0)\leq 0$ or $g_{k}(\sigma_{N_{k}}^{2})>0$ is true and therefore $\mathcal{F}_{k} = \mathcal{F}_{k1}\cup \mathcal{F}_{k2} \cup \mathcal{F}_{k3}$. \\
From Lemma \ref{thm:paramlemma}, $\bar{p}\in\mathcal{P}$ implies $d_{k1}+d_{k2}-\sigma_{N_{k}}^{2}e^{-2t_{k}}-\sigma_{N_{k}}^{2} \leq 0$ for $k=1,2$. However, $(d_{k1},d_{k2},t_{k})\in\mathcal{F}_{k3}$ implies $g_{k}(\sigma_{N_{k}}^{2})>0$. This means that $d_{k1}+d_{k2}-\sigma_{N_{k}}^{2}e^{-2t_{k}}-\sigma_{N_{k}}^{2}>0$. Therefore, $\bar{p}\in\mathcal{P}$ implies, $(d_{k1},d_{k2},t_{k})\notin\mathcal{F}_{k3}$. Therefore, 
\begin{equation*}
\bar{p}\in\mathcal{P} \Rightarrow (d_{k1},d_{k2},t_{k}) \in \mathcal{F}_{k1}\cup \mathcal{F}_{k2}, k=1,2.
\end{equation*}
\end{proof}

In order to show that the Gaussian scheme described in Section \ref{sec:ach} achieves the lower bound on the sum rate, we parametrize the achievable sum rate now. Define, 
\begin{align}
d_{k1}' &= \Var(X_{k}|U_{k1},S) = \frac{\sigma_{N_{k}}^{2}\sigma_{W_{k1}}^{2}}{\sigma_{N_{k}}^{2}+\sigma_{W_{k1}}^{2}} \label{eqn:dk1}\\
d_{k2}' &= \Var(X_{k}|U_{k2},S) = \frac{\sigma_{N_{k}}^{2}\sigma_{W_{k2}}^{2}}{\sigma_{N_{k}}^{2}+\sigma_{W_{k2}}^{2}}\label{eqn:dk2}\\
t_{k}' &= I(X_{k};U_{k1},U_{k2}|S) = \frac{1}{2}\log \frac{\sigma_{N_{k}}^{2}(\sigma_{W_{k1}}^{2}+\sigma_{W_{k2}}^{2}+2a_{k})+\sigma_{W_{k1}}^{2}\sigma_{W_{k2}}^{2}-a_{k}^{2}}{\sigma_{W_{k1}}^{2}\sigma_{W_{k2}}^{2}-a_{k}^{2}} \label{eqn:tk}.
\end{align}
We can rewrite the last equation above as 
\begin{equation}\label{eqn:tk1}
\frac{1}{\frac{\sigma_{N_{k}}^{2}e^{-2t_{k}'}}{1-e^{-2t_{k}'}}+a_{k}}=\frac{1}{\sigma_{W_{k1}}^{2}+a_{k}}+\frac{1}{\sigma_{W_{k2}}^{2}+a_{k}}.
\end{equation}
Let $\bar{p}'= (d_{11}',d_{12}',d_{21}',d_{22}',t_{1}',t_{2}')$ denote the parameters achieved by the Gaussian scheme. By definition of $(U_{11},U_{12},U_{21},U_{22})\in\mathcal{U}$, $\bar{p}' \in\mathcal{F}$. This means that the achievable parameters correspond to a Gaussian scheme that satisfies the distortion constraints (\ref{eqn:d1con}), (\ref{eqn:d2con}) and (\ref{eqn:d0con}). We use the definition of functions in (\ref{eqn:rparam}) and the parameters introduced above in the following lemma, relating them to the sum rate achievable by the Gaussian scheme. 

\begin{lemma}\label{thm:paramach}
For all $(U_{11},U_{12},U_{21},U_{22})\in\mathcal{U}$ and $\sigma_{Z_{k}}^{2}\geq 0$, $k\in\{1,2\}$, 
\begin{align}
&I(X_{1},X_{2};U_{11},U_{12},U_{21},U_{22})+I(U_{11},U_{21};U_{12},U_{22}) \nonumber\\
& = r_{1}(d_{11}',d_{12}',t_{1}',\sigma_{Z_{1}}^{2}) + r_{2}(d_{21}',d_{22}',t_{2}',\sigma_{Z_{2}}^{2})
 +I(U_{11};U_{12}|S,Y_{1})+I(U_{21};U_{22}|S,Y_{2})+\frac{1}{2}\log\frac{\sigma_{S}^{4}}{\delta_{1}\delta_{2}}, \label{eqn:ach_param}
\end{align}
where 
\begin{align*}
\frac{1}{\delta_{1}} &= \frac{1}{\sigma_{S}^{2}}+\frac{1}{\sigma_{N_{1}}^{2}}+\frac{1}{\sigma_{N_{2}}^{2}}-\frac{d_{11}'}{\sigma_{N_{1}}^{4}}-\frac{d_{21}'}{\sigma_{N_{2}}^{4}}\\
\frac{1}{\delta_{2}} &=
\frac{1}{\sigma_{S}^{2}}+\frac{1}{\sigma_{N_{1}}^{2}}+\frac{1}{\sigma_{N_{2}}^{2}}-\frac{d_{12}'}{\sigma_{N_{1}}^{4}}-\frac{d_{22}'}{\sigma_{N_{2}}^{4}},
\end{align*}
$Y_{1}=X_{1}+Z_{1}$ and $Y_{2}=X_{2}+Z_{2}$, $Z_{1}$ and $Z_{2}$ are independent of both $X_{1}$ and $X_{2}$ and Gaussian distributed with mean zero and variance $\sigma_{Z_{1}}^{2}$ and $\sigma_{Z_{2}}^{2}$ respectively. 
\end{lemma}
This lemma is proved in Appendix \ref{sec:paramach}. We now show that the Gaussian scheme achieves the lower bound on the sum rate corresponding to every point $\bar{p}\in\mathcal{P}$. We prove this through the following lemma. 

\begin{lemma}\label{thm:equivalence}
For every $\bar{p}\in\mathcal{P}$, there exists an achievable $\bar{p}'\in\mathcal{F}$ and $\sigma_{Z_{k}}^{2}\geq 0$, $k=1,2$, such that 
\begin{align*}
&r_{1}(d_{11},d_{12},t_{1},\sigma_{Z_{1}}^{2})+r_{2}(d_{21},d_{22},t_{2},\sigma_{Z_{2}}^{2})+\frac{1}{2}\log\frac{\sigma_{S}^{4}}{D_{1}D_{2}} \\&=  r_{1}(d_{11}',d_{12}',t_{1}',\sigma_{Z_{1}}^{2})+r_{2}(d_{21}',d_{22}',t_{2}',\sigma_{Z_{2}}^{2})+I(U_{11};U_{12}|S,Y_{1})+I(U_{21};U_{22}|S,Y_{2})+\frac{1}{2}\log\frac{\sigma_{S}^{4}}{\delta_{2}\delta_{1}}.
\end{align*}
\end{lemma}

\begin{proof}
The proof closely follows the discussion in Section 5 in \cite{Wang2007}. Let $\bar{p}=(d_{11},d_{12},d_{21},d_{22},t_{1},t_{2})\in\mathcal{P}$. Choosing $d_{kl}'=d_{kl}$ for $k=1,2$ and $l=1,2$, from (\ref{eqn:d1eq}) and (\ref{eqn:d2eq}), we know that
\begin{equation}\label{eqn:deltaeq}
\delta_{1}=D_{1} \quad \delta_{2}=D_{2}. 
\end{equation}
From Lemma \ref{thm:paramlem}, we know that $(d_{k1},d_{k2},t_{k})\in\mathcal{F}_{k1}\cup \mathcal{F}_{k2}$. By definition, $\mathcal{F}_{k1}\cap\mathcal{F}_{k2}=\phi$. We now consider two cases, $(d_{k1},d_{k2},t_{k})\in\mathcal{F}_{k1}$ and $(d_{k1},d_{k2},t_{k})\in\mathcal{F}_{k2}$. 

\subsection{Case 1: $(d_{k1},d_{k2},t_{k})\in \mathcal{F}_{k1}$}
Since $(d_{k1},d_{k2},t_{k})\in \mathcal{P}_{k1}$, $g_{k}(0)>0$ and $g_{k}(\sigma_{N_{k}}^{2})\leq 0$. Therefore, there exists an $a_{k}^{*}\in(0,\sigma_{N_{k}}^{2}]$ that solves $g_{k}(a_{k})=0$. We set $a_{k}=a_{k}^{*}$. Further, $d_{k1}'=d_{k1}$ and  $d_{k2}'=d_{k2}$ imply that $\sigma_{W_{k1}}^{2}=\alpha_{k1}$ and $\sigma_{W_{k2}}^{2}=\alpha_{k2}$. Therefore, we conclude from (\ref{eqn:tk1}) and $g_{k}(a_{k}^{*})=0$ that $t_{k}'=t_{k}$. We now need to show that this choice of $a_{k}^{*}$ is such that $\sigma_{W_{k1}}^{2}\sigma_{W_{k2}}^{2}\geq (a_{k}^{*})^{2}$. Since $\alpha_{k0}\geq 0$ and $a_{k}^{*}\in(0,\sigma_{N_{k}}^{2}]$, 
\begin{equation*}
\alpha_{k0}+a_{k}^{*}\geq a_{k}^{*} \Rightarrow \frac{1}{\alpha_{k0}+a_{k}^{*}} \leq \frac{1}{a_{k}^{*}}.
\end{equation*}
Since $g(a_{k}^{*})=0$, 
\begin{align*}
&\frac{1}{\sigma_{W_{k1}}^{2}+a_{k}^{*}}+\frac{1}{\sigma_{W_{k2}}^{2}+a_{k}^{*}} \leq \frac{1}{a_{k}^{*}}\\
\Rightarrow & \frac{1}{\sigma_{W_{k1}}^{2}+a_{k}^{*}} \leq \frac{1}{a_{k}^{*}}-\frac{1}{\sigma_{W_{k2}}^{2}+a_{k}^{*}}\\
\Rightarrow & \frac{1}{\sigma_{W_{k1}}^{2}+a_{k}^{*}} \leq \frac{\sigma_{W_{k2}}^{2}}{a_{k}^{*}(\sigma_{W_{k2}}^{2}+a_{k}^{*})}\\
\Rightarrow & \sigma_{W_{k1}}^{2}+a_{k}^{*} \geq \frac{(a_{k}^{*})^{2}}{\sigma_{W_{k2}}^{2}}+a_{k}^{*}\\
\Rightarrow & \sigma_{W_{k1}}^{2}\sigma_{W_{k2}}^{2} \geq (a_{k}^{*})^{2}.
\end{align*}

Moreover, trivially, 
\begin{equation*}
r_{k}(d_{k1},d_{k2},t_{k},\sigma_{Z_{k}}^{2}) =  r_{k}(d_{k1}',d_{k2}',t_{k}',\sigma_{Z_{k}}^{2})
\end{equation*}
Also, 
\begin{align*}
\Var(U_{kl}|S,Y_{k}) &= \sigma_{N_{k}}^{2}+\sigma_{W_{kl}}^{2}-\frac{\sigma_{N_{k}}^{4}}{\sigma_{N_{k}}^{2}+\sigma_{Z_{k}}^{2}}, l=1,2\\
\Cov(U_{k1},U_{k2}|S,Y_{k}) &= \sigma_{N_{k}}^{2}\left[\begin{array}{ccc} 1 & 1 \\ 1 & 1 \end{array}\right] + 
\left[\begin{array}{ccc} \sigma_{W_{k1}}^{2} & -a_{k} \\ -a_{k} & \sigma_{W_{k2}}^{2} \end{array}\right] -\frac{\sigma_{N_{k}}^{4}}{\sigma_{N_{k}}^{2}+\sigma_{Z_{k}}^{2}}\left[\begin{array}{ccc} 1 & 1 \\ 1 & 1 \end{array}\right]
\end{align*}
The off diagonal entries in $\Cov(U_{k1},U_{k2}|S,Y_{k})$ are zero if 
\begin{equation*}
\sigma_{N_{k}}^{2}-a_{k} = \frac{\sigma_{N_{k}}^{4}}{\sigma_{N_{k}}^{2}+\sigma_{Z_{k}}^{2}}. 
\end{equation*}
By choosing $\sigma_{Z_{k}}^{2}=\frac{a_{k}\sigma_{N_{k}}^{2}}{\sigma_{N_{k}}^{2}-a_{k}}$ in this case,
\begin{equation*}
\Var(U_{k1}|S,Y_{k})\Var(U_{k2}|S,Y_{k}) = |\Cov(U_{k1},U_{k2}|S,Y_{k})|
\end{equation*}
and $I(U_{k1};U_{k2}|S,Y_{k}) = 0$. Note that we are allowed to choose $\sigma_{Z_{k}}^{2}=\frac{a_{k}\sigma_{N_{k}}^{2}}{\sigma_{N_{k}}^{2}-a_{k}}$ since $a_{k}\in(0,\sigma_{N_{k}}^{2}]$ in this case. 

\subsection{Case 2: $(d_{k1},d_{k2},t_{k})\in \mathcal{F}_{k2}$}
In this case, we set $a_{k}=0$ in (\ref{eqn:tk}) and achieve the corresponding $t_{k}'$. Since, $d_{k1}'=d_{k1}$ and  $d_{k2}'=d_{k2}$, we have $\sigma_{W_{k1}}^{2}=\alpha_{k1}$ and $\sigma_{W_{k2}}^{2}=\alpha_{k2}$. It follows from (\ref{eqn:alpha}) and (\ref{eqn:tk1}) that 
\begin{equation*}
\frac{1}{\sigma_{N_{k}}^{2}}+ \frac{1}{\alpha_{k1}}+\frac{1}{\alpha_{k2}} = \frac{1}{\sigma_{N_{k}}^{2}e^{-2t_{k}'}}.
\end{equation*}
Since $g_{k}(0)\leq 0$, this implies that 
\begin{align*}
\frac{1}{\alpha_{k0}}&\leq \frac{1}{\alpha_{k1}}+\frac{1}{\alpha_{k2}}\\
\Rightarrow \frac{1}{\alpha_{k0}}+\frac{1}{\sigma_{N_{k}}^{2}}&\leq \frac{1}{\sigma_{N_{k}}^{2}}+ \frac{1}{\alpha_{k1}}+\frac{1}{\alpha_{k2}}\\
\Rightarrow \frac{1}{\sigma_{N_{k}}^{2}e^{-2t_{k}}}&\leq \frac{1}{\sigma_{N_{k}}^{2}}+ \frac{1}{\alpha_{k1}}+\frac{1}{\alpha_{k2}}\\
&=\frac{1}{\sigma_{N_{k}}^{2}e^{-2t_{k}'}}.
\end{align*}
Therefore, we get that $t_{k}\leq t_{k}'$. By achieving $t_{k}'$ instead of $t_{k}$, we still satisfy the central distortion constraint for the original problem and also ensure $(d_{k1}',d_{k2}',t_{k}')\in \mathcal{F}_{k}$. Further, we choose $\sigma_{Z_{k}}^{2}=0$ in this case. Therefore
\begin{equation*}
r_{k}(d_{k1},d_{k2},t_{k},0) = \frac{1}{2}\log\frac{\sigma_{N_{k}}^{4}}{d_{k1}d_{k2}} =  r_{k}(d_{k1},d_{k2},t_{k}',0),
\end{equation*}
Moreover, since $\sigma_{Z_{k}}^{2} = 0$ and $a_{k}=0$
\begin{equation*}
I(U_{k1};U_{k2}|S,Y_{k}) = I(U_{k1};U_{k2}|S,X_{k}) = 0.
\end{equation*}
\\
The lemma follows from the cases considered above. 
\end{proof}
Therefore, it follows from Lemma \ref{thm:lb} and Lemma \ref{thm:equivalence} that for every $\bar{p}\in\mathcal{P}$, there exists an achievable $\bar{p}'\in\mathcal{F}$ such that the sum rate achievable by the Gaussian scheme is equal to the lower bound on the sum rate. This proves the optimality of the Gaussian scheme for the sum rate of the vacationing-CEO problem. 

\section{Conclusion}
We introduced the vacationing-CEO problem which in essence, is a CEO problem with multiple descriptions. We described a Gaussian achievable scheme and presented a lower bound for the sum rate as an optimization problem over the code parameters. We also showed that the Gaussian scheme is optimal in terms of sum rate for a class of distortion constraints. Future work includes extending the result to other distortion regimes and considering a two terminal source coding problem with multiple descriptions. 

\appendices

\section{Proof of Lemma \ref{thm:achsum}}\label{sec:lemachsum}
In order to prove Lemma \ref{thm:achsum}, we need to show that $\forall \delta>0$, there exist $(R_{11}',R_{12}',R_{21}',R_{22}')$ and $(R_{11},R_{12},R_{21},R_{22})$ that satisfy (\ref{eqn:generrors}) and (\ref{eqn:decerrors}) such that 
\begin{equation*}
|R_{11}+R_{21}+R_{12}+R_{22} - I(X_{1},X_{2};U_{11},U_{12},U_{21},U_{22})-I(U_{11},U_{21};U_{12},U_{22})|\leq\delta. 
\end{equation*}
Let $\epsilon=\frac{\delta}{8}$ and $(U_{11},U_{12},U_{21},U_{22})\in\mathcal{U}$. We choose
\begin{align*}
R_{11}' &= I(X_{1};U_{11})+\epsilon  & R_{12}' &= I(X_{1};U_{12}|U_{11})+I(U_{11};U_{12})+\epsilon \\
R_{21}' &= I(X_{2};U_{21})+\epsilon  & R_{22}' &= I(X_{2};U_{22}|U_{21})+I(U_{21};U_{22})+\epsilon \\
R_{11} &= R_{11}'-I(U_{11};U_{21})+\epsilon &  R_{21} &= R_{21}'+\epsilon \\
R_{12} &= R_{12}'-I(U_{12};U_{22})+\epsilon &  R_{22} &= R_{22}'+\epsilon. 
\end{align*}
Note that $(R_{11}',R_{12}',R_{21}',R_{22}')$ satisfy (\ref{eqn:generrors}) and $(R_{11},R_{12},R_{21},R_{22})$ satisfy (\ref{eqn:decerrors}). Therefore,
\begin{align*}
R_{11}+R_{21}+R_{12}+R_{22} = & R_{11}'+R_{12}'+R_{21}'+R_{22}'-I(U_{11};U_{21})-I(U_{12};U_{22}) +4\epsilon\\
= & I(X_{1};U_{11},U_{12})+I(U_{11};U_{12})+I(X_{2};U_{21},U_{22})+I(U_{21};U_{22}) \\
&-I(U_{11};U_{21})-I(U_{12};U_{22})+8\epsilon\\
= &I(X_{1},X_{2};U_{11},U_{12},U_{21},U_{22})+I(U_{11},U_{21};U_{12},U_{22})+\delta. 
\end{align*}
Allowing $\delta\rightarrow 0$,  we see that the Gaussian scheme achieves the sum rate 
\begin{equation*}
\inf_{(U_{11},U_{12},U_{21},U_{22})\in\mathcal{U}}I(X_{1},X_{2};U_{11},U_{12},U_{21},U_{22})+I(U_{11},U_{21};U_{12},U_{22}). 
\end{equation*}

\section{Proof of Lemma \ref{thm:paramach}}\label{sec:paramach}
By procedural steps, we have 
\begin{align*}
I(X_{1},X_{2};U_{11},U_{21},U_{12},U_{22})+&I(U_{11},U_{21};U_{12},U_{22})\\&= I(S;U_{11},U_{21},U_{12},U_{22})+I(X_{1},X_{2};U_{11},U_{21},U_{12},U_{22}|S)\\
& \quad +I(U_{11},U_{21};U_{12},U_{22})\\
&=  I(S;U_{11},U_{21},U_{12},U_{22}) + I(X_{1};U_{11},U_{12}|S)+I(X_{2};U_{21},U_{22}|S)\\
& \quad +I(U_{11},U_{21};U_{12},U_{22})\\
&=  I(S;U_{11},U_{21},U_{12},U_{22}) + t_{1}'+t_{2}'+I(U_{11},U_{21};U_{12},U_{22}).
\end{align*}
Recall that $Y_{1}=X_{1}+Z_{1}$ and $Y_{2}=X_{2}+Z_{2}$ where $Z_{1}$ and $Z_{2}$ are Gaussians with mean zero and variance $\sigma_{Z_{1}}^{2}$ and $\sigma_{Z_{2}}^{2}$ and independent of $S$, $X_{1}$ and $X_{2}$. Now,
\begin{align}
I(U_{11},U_{21};U_{12},U_{22}) =& h(U_{11},U_{21})+h(U_{12},U_{22})-h(U_{11},U_{21},U_{12},U_{22})\nonumber\\
=& h(U_{11},U_{21})+h(U_{12},U_{22})-h(U_{11},U_{21},U_{12},U_{22})\nonumber\\
& -h(U_{11},U_{21}|S,Y_{1},Y_{2})-h(U_{12},U_{22}|S,Y_{1},Y_{2})\nonumber\\
& +h(U_{11},U_{21},U_{12},U_{22}|S,Y_{1},Y_{2}) + I(U_{11},U_{21};U_{12},U_{22}|Y_{1},Y_{2},S)\nonumber\\
=& I(S,Y_{1},Y_{2};U_{11},U_{21})+I(S,Y_{1},Y_{2};U_{12},U_{22})\nonumber\\
& -I(S,Y_{1},Y_{2};U_{11},U_{12},U_{21},U_{22})+I(U_{11};U_{12}|S,Y_{1})+I(U_{21};U_{22}|S,Y_{2}). \label{eqn:step11}
\end{align}
For $l=1,2$, let $\delta_{l}=\sigma_{S}^{2}e^{-2I(S;U_{1l},U_{2l})}$. Now, we can compute mutual information expressions between Gaussian random variables or use the fact that Gaussian random variables satisfy Lemma 3.1 in \cite{Prabhakaran2004} with equality to conclude that, 
\begin{align*}
\frac{1}{\sigma_{S}^{2}}e^{2I(S;U_{1l},U_{2l})} &= 
\frac{1}{\sigma_{S}^{2}} + \frac{1-e^{-2I(X_{1};U_{1l}|S)}}{\sigma_{N_{1}}^{2}} + \frac{1-e^{-2I(X_{2};U_{2l}|S)}}{\sigma_{N_{2}}^{2}}\\
\Rightarrow \frac{1}{\delta_{l}} &= \frac{1}{\sigma_{S}^{2}} + \frac{1}{\sigma_{N_{1}}^{2}} + \frac{1}{\sigma_{N_{2}}^{2}}-\frac{d_{1l}'}{\sigma_{N_{1}}^{4}}-\frac{d_{2l}'}{\sigma_{N_{2}}^{4}}. 
\end{align*}
Therefore, 
\begin{align}
I(S,Y_{1},Y_{2};U_{1l},U_{2l}) &= I(S;U_{1l},U_{2l})+I(Y_{1};U_{1l}|S)+I(Y_{2};U_{2l}|S)\nonumber\\
&= \frac{1}{2}\log \frac{\sigma_{S}^{2}(\sigma_{N_{1}}^{2}+\sigma_{Z_{1}}^{2})(\sigma_{N_{2}}^{2}+\sigma_{Z_{2}}^{2})}{\delta_{l}(d_{1l}'+\sigma_{Z_{1}}^{2})(d_{2l}'+\sigma_{Z_{2}}^{2})}. \label{eqn:step21}
\end{align}
Observe that
\begin{align}
I(S,Y_{1},Y_{2};U_{11},U_{12},U_{21},U_{22}) &= I(S;U_{11},U_{21},U_{12},U_{22})+I(Y_{1},Y_{2};U_{11},U_{21},U_{12},U_{22}|S)\nonumber\\
&= I(S;U_{11},U_{21},U_{12},U_{22})+I(Y_{1};U_{11},U_{12}|S)+I(Y_{2};U_{21},U_{22}|S), \label{eqn:step31}
\end{align}
and for $k=1,2$
\begin{align}
I(Y_{k};U_{k1},U_{k2}|S) &= -h(Y_{k}|S,U_{k1},U_{k2})+h(Y_{k}|S)\nonumber\\
&\overset{(a)}{=} -\frac{1}{2}\log(e^{\frac{2}{n}h(X_{k}|S,U_{k1},U_{k2})}+e^{\frac{2}{n}h(Z_{k})})+h(Y_{k}|S)\nonumber\\
&= -\frac{1}{2}\log(e^{\frac{2}{n}(h(X_{k}|S)-I(X_{k};U_{k1},U_{k2}|S))}+e^{\frac{2}{n}h(Z_{k})})+h(Y_{k}|S)\nonumber\\
&= -\frac{1}{2}\log(\sigma_{N_{k}}^{2}e^{-2t_{k}'}+\sigma_{Z_{k}}^{2})+\frac{1}{2}\log(\sigma_{N_{k}}^{2}+\sigma_{Z_{k}}^{2}), \label{eqn:step41}
\end{align}
where (a) follows from EPI for Gaussians. From (\ref{eqn:step11}), (\ref{eqn:step21}), (\ref{eqn:step31}) and (\ref{eqn:step41}),
\begin{align*}
I(U_{11},U_{21};U_{12},U_{22}) = & \frac{1}{2}\log\frac{(\sigma_{N_{1}}^{2}+\sigma_{Z_{1}}^{2})(\sigma_{N_{2}}^{2}+\sigma_{Z_{2}}^{2})\sigma_{S}^{4}}{(d_{12}'+\sigma_{Z_{1}}^{2})(d_{22}'+\sigma_{Z_{2}}^{2})(d_{11}'+\sigma_{Z_{1}}^{2})(d_{21}'+\sigma_{Z_{2}}^{2})\delta_{1}\delta_{2}}\\
& + \frac{1}{2}\log(\sigma_{N_{1}}^{2}e^{-2t_{1}'}+\sigma_{Z_{1}}^{2})(\sigma_{N_{2}}^{2}e^{-2t_{2}'}+\sigma_{Z_{2}}^{2})-I(S;U_{11},U_{21},U_{12},U_{22})\\
& +I(U_{11};U_{12}|S,Y_{1})+I(U_{21};U_{22}|S,Y_{2})
\end{align*}
and 
\begin{align*}
&I(X_{1},X_{2};U_{11},U_{12},U_{21},U_{22})+I(U_{11},U_{21};U_{12},U_{22}) \nonumber\\
& = r_{1}(d_{11}',d_{12}',t_{1}',\sigma_{Z_{1}}^{2}) + r_{2}(d_{21}',d_{22}',t_{2}',\sigma_{Z_{2}}^{2})
 +I(U_{11};U_{12}|S,Y_{1})+I(U_{21};U_{22}|S,Y_{2})+\frac{1}{2}\log\frac{\sigma_{S}^{4}}{\delta_{1}\delta_{2}}.
\end{align*}


\end{document}